\documentclass[runningheads]{llncs}
\usepackage[T1]{fontenc}
\usepackage{graphicx}
\usepackage{amssymb}
\usepackage{physics}
\usepackage{algorithm}
\usepackage{ wasysym }
\usepackage{algpseudocode}
\newcommand{\Endproof}{\hfill$\Box$\\}
\begin{document}
\title{Circuits of Quantum Hashing and Quantum Fourier Transform for a Cactus as a Qubit Connectivity Graph\thanks{The research has been supported by Russian Science Foundation Grant 24-21-00406, \url{ https://rscf.ru/en/project/24-21-00406/}. The study in Section 3.2 was performed under the development programme of the Volga Region Mathematical Center (agreement no. 075-02-2026-1328).}}
%
\titlerunning{Circuits of QH and QFT for a Cactus as a Qubit Connectivity Graph}
%
\author{
Kamil Khadiev\inst{1,2}\orcidID{0000-0002-5151-9908}\and Ilnur Valeev\inst{3}}
\authorrunning{K.Khadiev and I. Valeev}
%
\institute{Institute of Computational Mathematics and Information Technologies, Kazan Federal University, Kazan, Tatarstan, Russia \and
Zavoisky Physical-Technical Institute,
FRC Kazan Scientific Center of RAS, Kazan, Tatarstan, Russia\and
N.I. Lobachevsky Institute of Mathematics and Mechanics, Kazan Federal University, Kazan, Tatarstan, Russia\\
\email{kamilhadi@gmail.com}}
\maketitle
\begin{abstract}
We present a quantum circuit implementation of the quantum hashing algorithm (quantum fingerprinting) for a quantum device with restrictions on the application of two-qubit gates by a qubit connectivity graph. We present an optimization technique for the shallow circuit for quantum hashing in the case of a cactus as a qubit connectivity graph. The algorithm has $O(n^3)$ complexity to build the circuit, where $n$ is the number of qubits and $m$ is the number of connections (edges) in the graph. It is improvement compared to the existing exponential-time algorithm in the case of arbitrary graphs.  The algorithm uses solution for the shortest non-simple 1-covering path problem as a subroutine. We present an $O(n^3)$-time solution for this graph-theory problem in the case of a cactus. This result can be interesting independently. The algorithm also used for improving of the quantum circuit for Quantum Fourier Transform.

%

\end{abstract}
%
%
\section{Introduction}
Quantum computing \cite{nc2010,aazksw2019part1} is one of the hot topics in computer science in recent decades.
%
%
One of the techniques that allows us to develop space-efficient quantum algorithms is quantum fingerprinting or quantum hashing \cite{akvz2025}. 
This technique is well known, and it allows us to compute a short hash or fingerprint that identifies a (potentially long) string of data with high probability.
The probabilistic (randomized) technique was developed in \cite{Fre79}. Then, its quantum counterpart for automata was developed that was improved by Ambainis and Nahimovs in \cite{an2009}. Later, an explicit definition of the quantum fingerprinting to build an efficient quantum communication protocol for equality testing was provided \cite{bcwd2001}. The technique was applied for branching programs (non-uniform automata-like model) 
\cite{agkmp2005}. 
Later, the concept of cryptographic quantum hashing was developed
\cite{akvz2025}.
The technique was also extended for qudits \cite{av2022}.
This approach was widely used in different areas such as stream processing algorithms \cite{l2009}, query model algorithms
\cite{aaksv2022,asa2024}, 
online algorithms \cite{kk2019disj,kk2022,kkkmry2023,kkzmkry2022,kk2019online2w,kk2019,kkm2018}, branching programs \cite{aaksv2022,kkk2022}, 
developing quantum devices \cite{v2016model}, automata (discussed earlier, \cite{gy2017}),
machine learning  \cite{zfsd2025} and 
others.

The technique was implemented in a photon-based real quantum device \cite{avsak2022,tavak2021,phyf2022,zllzh2023}. At the same time, in all these implementations, the algorithm was embedded in the devices. If we consider ``universal'' quantum devices such as IBMQ quantum computers or similar, it is important to minimize the number of quantum gates in the circuit implementation of the algorithm with respect to the architecture restrictions.  Many types of quantum computers (for example, quantum devices based on superconductors) do not allow us to apply two-qubit gates to an arbitrary pair of qubits. They have a specific architecture of qubit connectivity that is represented by a qubit connectivity graph. Vertices of the graph correspond to qubits, and two-qubit gates can be applied only to qubits corresponding to vertices connected by an edge. 

In this paper, we are interested in minimizing the number of CNOT gates which we call the CNOT cost of the quantum circuit. The CNOT gate is a two-qubit gate that is a quantum analog of ``exclusive or'' operation for classical computation. This gate is the most difficult for physical implementation among the gates used for quantum fingerprinting. 
The standard implementation of quantum hashing circuits is a uniformly controlled rotation operator.
The CNOT cost (the number of CNOT gates) of this operator with respect to a qubit connectivity graph was previously discussed in \cite{bvms2005mottonen2006decompositions} 
for the linear nearest-neighbor (LNN) architecture and in \cite{zkk2023,zk2025icmne} for several examples of more complex graphs. The issue with the circuit is an exponential number of CNOT gates with respect to the number of qubits.
%
K\={a}lis \cite{kalis18} suggested a shallow circuit for an automaton with 3 qubits recognizing the unary language $MOD_p=\{a^\ell: \ell $ mod $p=0\}$ that is simple enough, but has all the required properties to demonstrate the quantum hashing (fingerprinting) algorithm. Later, in a general way, the approach was developed in \cite{ziiatdinov2023gaps,zkk2025}. This approach allows us to reduce the CNOT cost. Computational experiments show that we can find the required parameters of the hash function with the same number of qubits as for the standard circuit for the quantum hashing algorithm.  Therefore, the CNOT cost of the shallow circuit can be exponentially less than for the standard circuit. At the same time, the theoretical exponential superiority of the shallow circuit is not shown \cite{ziiatdinov2023gaps,zkk2025}. By the way, the method is very promising for current and near-future quantum devices that definitely cannot support a huge number of CNOT gates. The efficient implementation of a shallow circuit for an automaton that recognizes the $MOD_p$ language was discussed in \cite{ksy2024} for devices based on the LNN architecture, and for more complex architecture \cite{k2024aliya}, that is a cycle with tails (like a ``sun'' and  ``two joint suns'') represented by 16-qubit and 27-qubit Eagle r3 IBMQ architectures. A method to construct a circuit for an arbitrary connected graph was suggested \cite{kkcw2025}, but this method builds the circuit with exponential time complexity.
Vasiliev \cite{v2023} discussed a similar method, but for Rz gates instead of Ry gates. The effectiveness of the method was shown in computational experiments in \cite{kmak2025}.
In this paper, we present a method for building a quantum circuit for the quantum hashing algorithm for a wide class of qubit connectivity graphs. The class is cacti that are connected graphs in which any two simple cycles have at most one vertex in common. The CNOT cost of the constructed circuit is 
between $2n\ell -4\ell +2$ and $6n\ell -7\ell+2$ depending on the complexity of the graph,
where $n$ is the number of qubits and $\ell$ is the length of the input word.
Our algorithm uses a solution of the shortest non-simple covering path problem, which is a modification of the shortest covering path problem \cite{cpr1994}. We present a polynomial-time solution for the problem in the case of a cactus that can be independently interesting to the reader. The time complexity of the solution for the shortest non-simple covering path problem and constructing the circuit is $O(n^3)$, where $n$ is the number of vertices and $m$ is the number of edges.
The polynomial time complexity is an important advantage compared to existing methods \cite{kkcw2025} that have exponential time complexity $O((n+m)2^n)$. So, our new result allows us to use the method for practical important qubit connectivity graphs with $n\geq 40$. The situation when an NP-problem has a polynomial solution for cacti is not unique. An example of such a problem is the shortest path cover problem \cite{bdfp2007}. At the same time, this is a different problem and has another algorithm for cacti compared to the one presented in this paper.  

Another algorithm that has similar quantum circuit is the Quantum Fourier Transform (QFT) \cite{k1995}.
It is used in quantum addition \cite{d2000}, quantum phase estimation (QPE) \cite{k1995}, quantum amplitude estimation (QAE) \cite{bhmt2002}, the algorithm for solving linear systems of equations \cite{hhl2009}, Shor’s factoring algorithm \cite{s1999}, and others.
The closeness of QFT and quantum hashing and, particularly, their quantum circuits were discussed in \cite{k2024aliya,kkcw2025,av2020}.
The quantum circuit of QFT for specific graphs was discussed in \cite{park2023reducing} for LNN, in \cite{k2024aliya} for a cycle with tails (like a ``sun'' or ``two joint suns''), and in \cite{kkcw2025,kksk2026} for arbitrary qubit connectivity graphs.
The best result that can be applied for cacti is \cite{kksk2026} that uses the path similar to the shortest non-simple covering path problem. At the same time, the algorithm from \cite{kksk2026} has exponential time complexity $O((m+n)2^n)$ or polynomial time complexity for an approximate algorithm. So, our result uses the same algorithm to generate the circuit using the path. At the same time, our algorithm for a cactus is exact (not approximate) and has polynomial time complexity.

The structure of this paper is the following.
Section \ref{sec:prelims} describes the required notations and preliminaries. 
Section \ref{sec:shallow} provides an algorithm for generating a shallow quantum circuit for quantum hashing. The QFT method is discussed in Section \ref{sec:qft}.
Graph theory tools are presented in Section \ref{sec:tools}. 
The final Section \ref{sec:conclude} concludes the paper and contains some open questions.

\section{Preliminaries}\label{sec:prelims}
\textbf{Graph Theory.}
Let us consider an undirected unweighted graph $G=(V,E)$, where $V$ is the set of vertices and $E$ is the set of undirected edges. Let $n=|V|$ be the number of vertices, and $m=|E|$ be the number of edges. 

A non-simple path $P$ is a sequence of vertices $(v_{i_1},\dots,v_{i_h})$ that are connected by edges, that are $(v_{i_j},v_{i_{j+1}})\in E$ for all $j\in\{1,\dots,h-1\}$. Note that a non-simple path can contain duplicates.
Let the length of the path be the number of edges in the path, $len(P)=h-1$.
A path $P=(v_{i_1},\dots,v_{i_h})$ is called simple if there are no duplicates among $v_{i_1},\dots,v_{i_h}$. 
The distance $dist(v,u)$ is the length of the shortest path between vertices $v$ and $u$. Typically, when we say just a ``path'', we mean a ``non-simple path''. 
Let $\textsc{Neighbors}(v)$ be a list of neighbours for a vertex $v$, i.e., $\textsc{Neighbors}(v)=(u_{i_1},\dots,u_{i_k})$ such that $(v,u_{i_j})\in E$, and $|\textsc{Neighbors}(v)|=k$ is the length of the list.

 {\em A cactus} (or a cactus tree or a cactus graph) is a connected graph in which any two simple cycles have at most one vertex in common. Equivalently, it is a connected graph in which every edge belongs to at most one simple cycle, or (for nontrivial cacti) in which every block (maximal subgraph without a cut-vertex) is an edge or a cycle. {\em A vertex cactus} (or a vertex cactus tree or a vertex cactus graph) is a connected graph in which any vertex belongs to at most one simple cycle.

Let us discuss the ``Shortest non-simple 1-covering path'' problem that is a modification of the well-known ``Shortest covering path'' problem (SCPP or SCP problem) \cite{cpr1994}. 
The ``Shortest non-simple 1-covering path'' problem (1-SNSCPP or 1-SNSCP problem) is defined as follows. Let $P=(v_{i_1},\dots,v_{i_k})$ be a \textbf{non-simple}  path. We say that the path covers all visiting vertices and vertices that are connected with visited vertices by one edge. Formally, the path $P$ covers a set of vertices $R_P$ such that any vertex $v$ from this set is either
    (i) $v$ belongs to $P$ (there is $j\in\{1,\dots,k\}$ such that $v=v_{i_j}$);
    (ii) $v$ is connected with a vertex from $P$ (there is $j\in\{1,\dots,k\}$ such that $(v,v_{i_j})\in E$).
Let $B_P=R_P\backslash\{v_{i_1},\dots,v_{i_k}\}$, i.e. they are vertices connected with visited vertices by one edge. 
If the path $P$ covers all vertices ($R_P=V$), then we call it a 1-covering path. The solution of the 1-SNSCP problem is the shortest 1-covering path.

Let us consider an undirected weighted graph $S=(V',E')$, where $V'$ is the set of vertices, and $E'$ is the set of undirected weighted edges. Let $w:V'\times V'\to \mathbb{R}$ be a weight function. We assume that $w(v,u)=\infty$ if $(v,u)\not\in E$.

For a path $P=(v_{i_1},\dots,v_{i_h})$, the weight $w(P)$ is the sum of edge weights in the path, i.e. $w(P)=\sum_j^{h-1}w(v_{i_j},v_{i_{j+1}})$.

\textbf{Quantum circuits.}
Quantum circuits consist of qubits and a sequence of gates applied to these qubits. A state of a qubit is a column-vector from ${\cal H}^2$ Hilbert space. It can be represented by $a_0|0\rangle+a_1|1\rangle$, where $a_0,a_1$ are complex numbers such that $|a_0|^2+|a_1|^2=1$, and $|0\rangle$ and $|1\rangle$ are unit vectors. Here we use the Dirac notation. A state of $n$ qubits is represented by a column-vector from ${\cal H}^{2^n}$ Hilbert space. It can be represented by $\sum_{i=0}^{2^n-1}a_i|i\rangle$, where $a_i$ is a complex number such that $\sum_{i=0}^{2^n-1}|a_i|^2=1$, and $|0\rangle,\dots |2^n-1\rangle$ are unit vectors. Graphically, on a circuit, qubits are presented as parallel lines. 

As basic gates, we consider the following ones:

  $H=\frac{1}{\sqrt{2}}\begin{pmatrix}
1 & 1 \\
1 & -1 
\end{pmatrix}$,     $X=\begin{pmatrix}
0 & 1 \\
1 & 0 
\end{pmatrix}$,
 $R_y(\xi)=\begin{pmatrix}
cos(\xi/2) & -sin(\xi/2) \\
sin(\xi/2) & cos(\xi/2) 
\end{pmatrix}$,

$R_z(\xi)=\begin{pmatrix}
e^{\frac{i\xi}{2}} & 0 \\
0 & e^{-\frac{i\xi}{2}} 
\end{pmatrix}$,
$CNOT=\begin{pmatrix}
1 & 0 & 0 & 0\\
0 & 1 & 0 & 0\\
0 & 0 & 0 & 1\\
0 & 0 & 1 & 0 
\end{pmatrix}$.  

Additionally, we consider four non-basic gates. The first three gates are

\noindent
$CR_y(\xi)=\begin{pmatrix}
1 & 0 & 0 & 0\\
0 & 1 & 0 & 0\\
0 & 0 & cos(\xi/2) & -sin(\xi/2) \\
0 & 0 & sin(\xi/2) & cos(\xi/2) 
\end{pmatrix}$,
$CR_z(\xi)=\begin{pmatrix}
1 & 0 & 0 & 0\\
0 & 1 & 0 & 0\\
0 & 0 & e^{\frac{i\xi}{2}}  & 0\\
0 & 0 & 0 & e^{-\frac{i\xi}{2}} 
\end{pmatrix}$, 

$SWAP=\begin{pmatrix}
1 & 0 & 0 & 0\\
0 & 0 & 1 & 0\\
0 & 1 & 0 & 0\\
0 & 0 & 0 & 1 
\end{pmatrix}$,  $R_k=\begin{pmatrix}
1 & 0 \\
0 & e^{\frac{i\pi}{2^{k-1}}} 
\end{pmatrix}$,
$CR_k=\begin{pmatrix}
1 & 0 & 0 & 0\\
0 & 1 & 0 & 0\\
0 & 0 & 1 & 0\\
0 & 0 & 0 & e^{\frac{i\pi}{2^{k-1}}} 
\end{pmatrix}$.

The fourth is a uniformly controlled rotation gate on $n-1$ control qubits and one target qubit. We denote it by $UCR^{n-1}_b$ gate. Here, $b$ can be $y$ or $z$. Let $|\psi\rangle$ be an $(n-1)$-qubit quantum control register, and $|\phi\rangle$ be one target qubit. The gate applies the $R_a(\xi_i)$ rotation gate to $|\phi\rangle$ if the control qubit register is $|i\rangle$ for $i\in\{0,\dots, 2^{n-1}-1\}$. If $a=y$, then the gate applies $R_y(\xi_i)$.  If $a=z$, then the gate applies $R_z(\xi_i)$.

%

The reader can find more information about quantum circuits in \cite{nc2010,aazksw2019part1,k2022lecturenotes}.

\section{Method for Constructing a Circuit for Quantum Hashing}\label{sec:hash}

In this section, we present a method that allows us to construct a circuit for the quantum fingerprinting or quantum hashing algorithm for a cactus $G=(V,E)$ that is a qubit connectivity graph for a device. 

A simple example of application of the quantum hashing algorithm is the quantum automaton for the unary $MOD_p=\{a^\ell: \ell$ mod $p=0$, $p$ is prime$\}$ language. The automata and more information about the quantum hashing (quantum fingerprinting) algorithm can be found in Appendix \ref{apx:hash}.

The $n$-qubit quantum circuit for the quantum hashing algorithm contains three main parts. We assume that we have qubits $q_0,\dots,q_{n-1}$.
The first part is to apply the Hadamard transformation $H$ to qubits with indexes $1,\dots, n-1$. The second part is to apply the main $U_a$ transformation of the quantum hashing algorithm. Often, the main transformation is applied several times (let it be $\ell$ times). In the case of quantum automaton for $MOD_p$, for the input word $a^\ell$, we apply the transformation $\ell$ times. The final part is to apply the Hadamard transformation $H$ to qubits with indexes $1,\dots, n-1$, and to measure all qubits.
In the next section, we discuss a shallow circuit for the $U_a$ transformation.

\subsection{Shallow Circuit for Quantum Hashing.}\label{sec:shallow}
Here we consider a shallow circuit \cite{kalis18,ziiatdinov2023gaps,zkk2025} for $n-1$ control qubits. The circuit can give up to exponential advantage in the number of CNOT gates comparing to the standard circuit based on the uniformly controlled rotation gate. Computational experiments shows that we can find the required parameters for hashing. At the same time, there is no proof the exponential advantage in theory \cite{zkk2025}. The implementation of the standard circuit and the shallow circuit for noisy simulators of real quantum devices show that the error due to the large number of CNOTS for the standard circuit is larger than the error due to required dependence of angels for the shallow circuit \cite{zkk2025}. So, the shallow circuit is more suitable for near-future quantum devices.
If we do not have restrictions for applying two-qubit gates (when $G$ is a complete graph or a ``star'' graph), then the shallow circuit for the main transformation of the quantum hashing algorithm is presented in Figure \ref{fig:qf}. In the case of the shallow circuit, we call the transformation $U_s$. Here we assume that we have qubits $q_1,\dots,q_{n}$ such that $q_1,\dots,q_{n-1}$ are control ones and $q_{n}$ is the target one.

\begin{figure}[!ht]
\begin{center}
\includegraphics[width=0.27\textwidth]{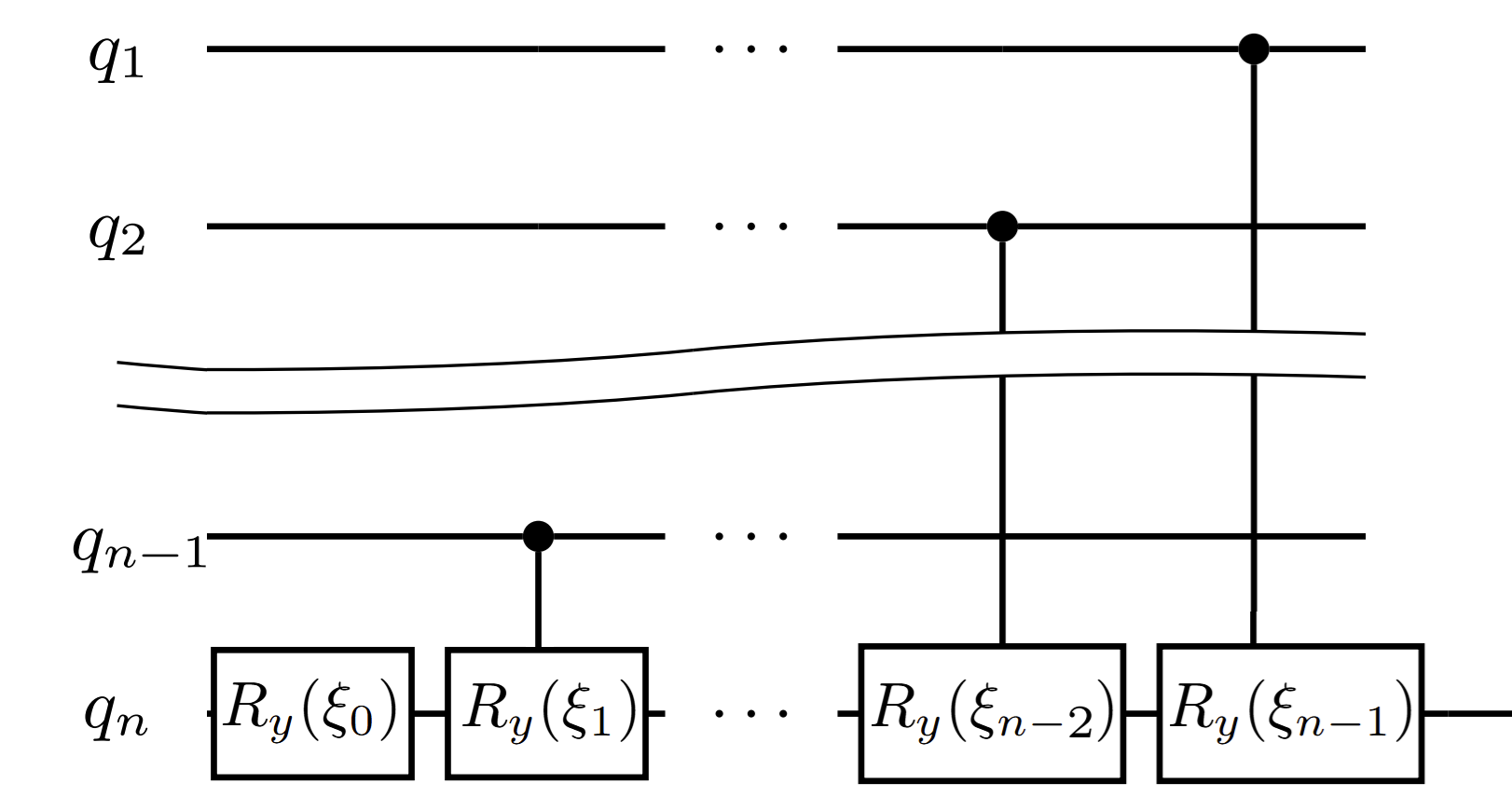}
\caption{\label{fig:qf} Shallow circuit for quantum fingerprinting or quantum hashing algorithm}    
\end{center}
\end{figure}

Let us present the main idea of the algorithm.
Let $P$ be the shortest non-simple 1-cover path for the cactus $G$ that is the solution of 1-SNSCPP. (See the definition of 1-SNSCPP in Section \ref{sec:prelims}). Initially, we associate the target qubit with the first vertex of the path $P$; and control qubits with other vertices of the graph. Because we can apply two-qubit gates only for adjacent vertices, the procedure moves the target qubit by the path $P$ from the first vertex of the $P$ to the last one. We move the target qubit using the SWAP gate. During the ``journey'' of the target qubit, we apply the control rotation operator to each neighbor vertex. Because the path $P$ covers all the vertices of the graph, it means that during the ``journey'', we can apply the control rotation operator to each edge of the graph (or each control qubit). So, this strategy allows us to implement the circuit.

In Section \ref{sec:tools} we present the algorithm for solving the 1-SNSCP problem in the case of a cactus. Here we assume that we have a procedure $\textsc{Shortest1CP}(G)$ that finds the shortest non-simple 1-cover path for the cactus $G$. Let us present a detailed description of the algorithm.

Let $q_i$ be the qubits of the original circuit or we call them ``original qubits''. Let the qubits of a physical device be associated with the vertices of the graph $G$, and we call them $v_i$.
The algorithm for constructing a circuit is the following.

    \textbf{Step 1.} We find the shortest 1-covering path in the cactus $G$ using the algorithm from Section \ref{sec:tools}. Assume that the path is $P=(v_{i_1}, \dots,v_{i_k})$.
    
    \textbf{Step 2.} The target qubit $q_{n}$ corresponds to the vertex $v_{i_1}$. The control qubits $q_1,\dots,q_{n-1}$ correspond to other vertices. Let us have a set $U$ of control qubits that are already used. Initially, it is empty $U\gets\emptyset$. Let $j$ be the index of the element in the path $P$ that corresponds to the target qubit. Initially, $j\gets 1$.
    
     In the next steps, the target qubit travels along the path $P$. 
    
    \textbf{Step 3.} We apply the control rotation with the control $v'$ and the target $v_{i_j}$ qubits, where $v'\in \textsc{Neighbors}(v_{i_j})\backslash\{v_{i_{j+1}}\}$ and $v'\not\in U$. In other words, $v'$ is a neighbor of $v_{i_j}$ but not $v_{i_{j+1}}$, and $v'$ is not visited. Then, we add $v'$ to the set $U$, i.e. $U\gets U\cup\{v'\}$.
   
    \textbf{Step 4.} If $v_{i_{j+1}}\not\in U$, then we apply a control rotation to the control $v_{i_{j+1}}$ and the target $v_{i_j}$ qubits. After that, we add $v_{i_{j+1}}$ to the set $U$, i.e. $U\gets U\cup\{v_{i_{j+1}}\}$. 
    
     \textbf{Step 5.} We apply the SWAP gate to $v_{i_{j}}$ and $v_{i_{j+1}}$. After that, we update $j\gets j+1$ because the value of the target qubit moves to $v_{i_{j+1}}$. If $j<k$, then we go to Step 3, and go to Step 6 otherwise.
     
     \textbf{Step 6.} We do this step if $j=k$. We apply control rotation with the control $v'$ and the target $v_{i_k}$ qubit, where $v'\in \textsc{Neighbors}(v_{i_k})$ and $v'\not\in U$. Then, we add $v'$ to the set $U$, i.e. $U\gets U\cup\{v'\}$. This is the final step of the algorithm.  

The implementation of the algorithm is presented in Appendix \ref{apx:algo-impl}. Let us discuss the CNOT cost of the generated circuit for $\ell$ applications of the main quantum hashing transformation $U_s$ (as an example, it is the main part of the circuit implementation of the automaton for $MOD_p$ processing $a^\ell$). For an odd application of $U_s$ we apply the presented algorithm. For an even application, we do the same algorithm, but move the target qubit in the reverse order from $v_k$ to $v_1$. The vertices $v'\in \textsc{Neighbors}(v_{i_k})$ are also considered in the reverse order. So, we have two sequential $\textsc{cR}(u,v_{i_k})$ gates, where  $\textsc{cR}(u,v_{i_k})$ represents the $CR_y(\xi)$ gate with $u$ as the control, and $v_{i_k}$ as the target qubits, and $\xi$ is the angle corresponding to $u$. Here $u$ is the last vertex in the list $\textsc{Neighbors}(v_{i_k})$. So, we can merge these two operators and use only one with the double angle. We have a similar situation with the first element from $\textsc{Neighbors}(v_{i_1})$.

The control rotation operator $\textsc{cR}(u,v)$ with the corresponding angle $\xi$ can be presented as a sequence of $\textsc{R}(v,\xi/2)$, $\textsc{cnot}(u,v)$, $\textsc{R}(v,-\xi/2)$, and $\textsc{cnot}(u,v)$, (Figure \ref{fig:cr}, the left circuit). Here $\textsc{R}(v,\xi/2)$ is the $R_y(\xi/2)$ gate applied to the qubit $v$; $\textsc{cnot}(u,v)$ is the gate $CNOT$ for $u$ as a control and $v$ as a target. The $\textsc{swap}(u,v)$ gate can be represented as $\textsc{cnot}(u,v)$, $\textsc{cnot}(v,u)$, and $\textsc{cnot}(u,v)$ (Figure \ref{fig:cr}, the middle circuit). 
Two sequential operators $CR_y$ and $SWAP$ that are $\textsc{cR}(u,v)$ and $\textsc{swap}(u,v)$ procedures can be represented by a circuit in Figure \ref{fig:cr} (the right circuit). Note that two sequential $\textsc{cnot}(u,v)$ gates are annihilated.

\begin{figure}[!ht]
\begin{center}
\includegraphics[width=0.35\textwidth]{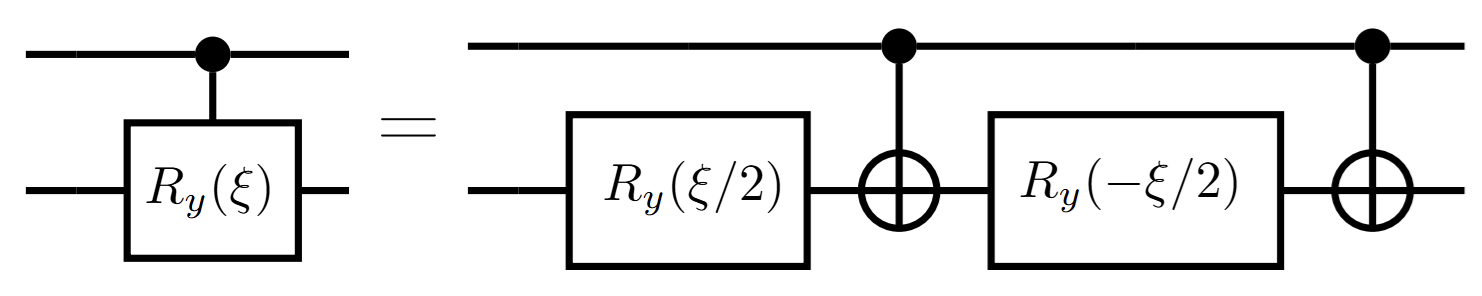}
$\quad$
\includegraphics[width=0.22\textwidth]{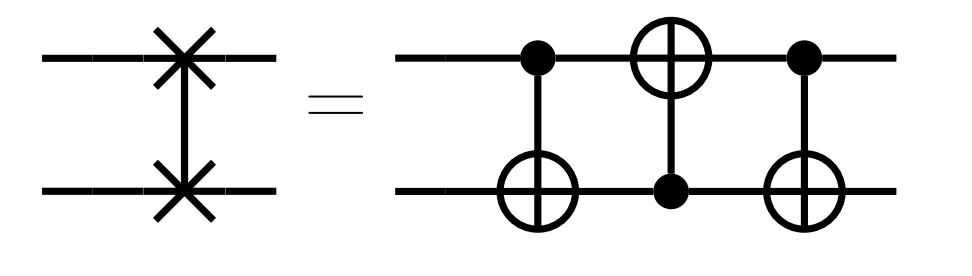}
$\quad$
\includegraphics[width=0.35\textwidth]{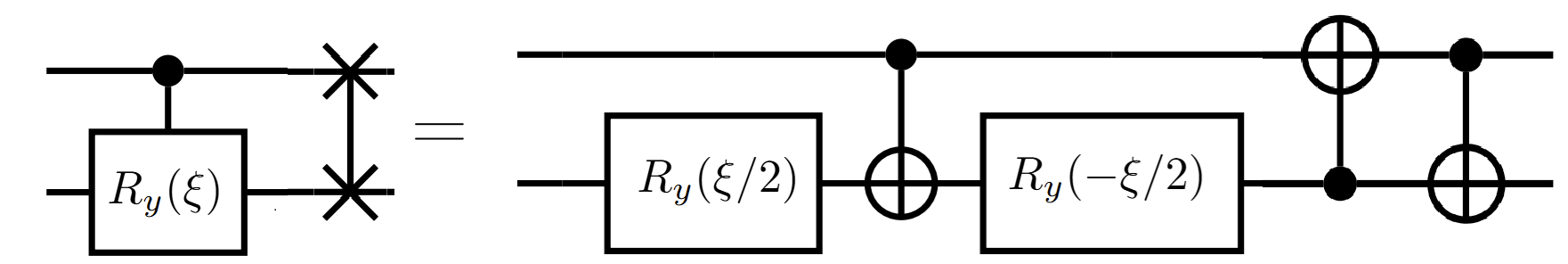}
\caption{\label{fig:cr} Representation of $CR_y$,  $SWAP$ gates, and a pair $CR_y$ and $SWAP$ gates using only basic gates.}
\end{center}
\end{figure}

Let us look at the CNOT cost of these operators, that is, the number of CNOT gates in the circuit.
We can say that CNOT cost of $\textsc{cR}(u,v)$ is $2$; of $\textsc{swap}(u,v)$ is $3$; of a pair $\textsc{cR}(u,v)$ and $\textsc{cnot}(u,v)$ is also $3$. 
Finally, we can discuss the CNOT cost of the constructed circuit for $\ell$ applications of the quantum hashing operator $U_s$.
\begin{theorem}\label{th:qh1}
    The CNOT cost of the circuit for $\ell$ applications of quantum hashing operator $U_s$ generated by the presented algorithm is $(3k + 2(n-k'))\ell-5\ell + 2$, where $k$ is the length of the 1-covering path $P=(v_{i_1},\dots,v_{i_k})$, and $k'$ is the number of distinct vertices in $P$, and $n$ is the number of vertices in the qubit connectivity graph. 
    (See Appendix \ref{apx:qh1})
\end{theorem}

Let us estimate the minimal and maximal possible CNOT cost.

\begin{corollary}\label{cr:path}
The CNOT cost of the circuit for $\ell$ applications of the quantum hashing operator $U_s$ generated by the presented algorithm is between $2n\ell-4\ell+2$ and $6n\ell-7\ell+2$, where $k$ is the length of the 1-covering path $P=(v_{i_1},\dots,v_{i_k})$, and $k'$ is the number of distinct vertices in $P$. 
(See Appendix \ref{apx:path})
\end{corollary}

\subsubsection{Comparing with Existing Results}\label{sec:compare-qh}
Let us compare our algorithm for a cactus  with existing technique for arbitrary graph \cite{kkcw2025}. Let us look to a simple cactus from Figure \ref{fig:cactus}.
\begin{figure}[!ht]
\begin{center}    \includegraphics[width=0.2\textwidth]{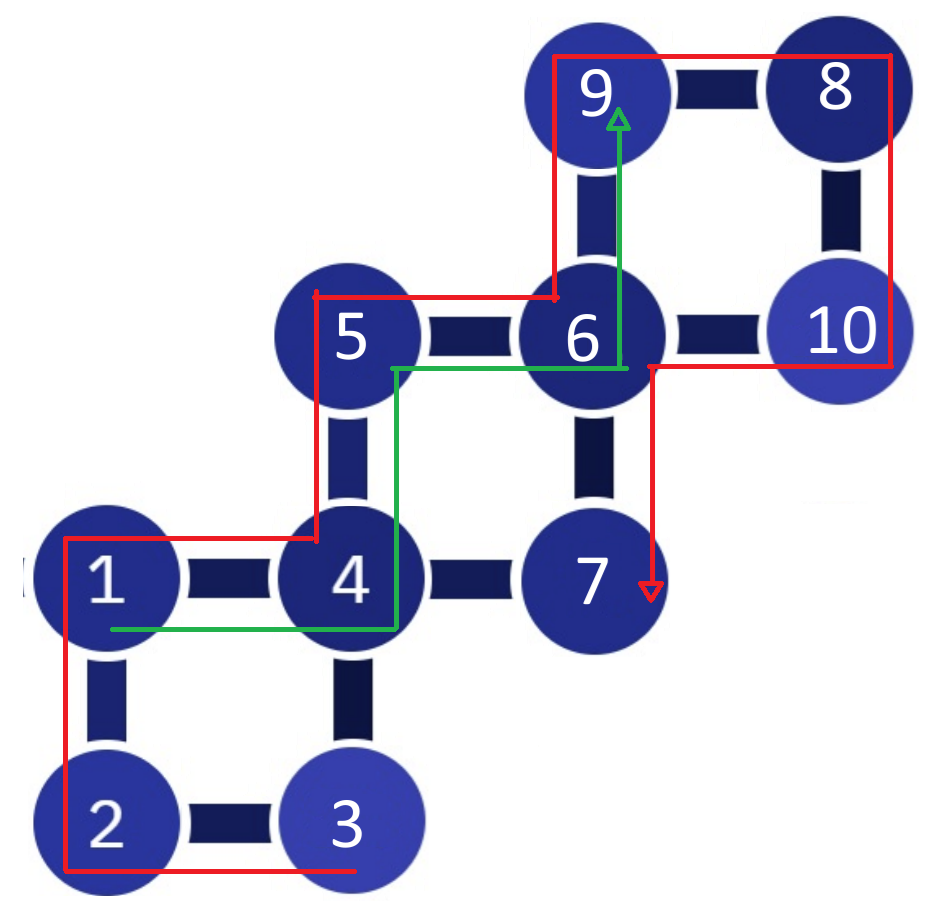}
\end{center}
\caption{\label{fig:cactus} A cactus as a qubits connectivity graph. The red path is the shortest path that visits all vertices at least once. It is used by \cite{kkcw2025}. The green path is the solution for the 1-SNSCP problem. It is used by our algorithm.}
\end{figure}

The algorithm of \cite{kkcw2025} uses the shortest path $P_1$ that visits all vertices at least once, which is presented as a red path on the figure. The length of the path is $k=10$ and the CNOT cost of one application of the quantum hashing algorithm is $3(k-2)+2\cdot 2=(10-2)\cdot 3+ 2\cdot 2=28$. 

Our algorithm uses a solution $P_2$ for the 1-SNSCP problem, which is presented as a green path on the figure. The length of the path is $4$ and there are $5$ vertices connected with vertices from the path. The CNOT cost of one application of the quantum hashing algorithm is $4\cdot 3+ 5\cdot 2=22$. 

If we consider a cactus that is a chain of cycles of size $4$ of a form similar to the Figure \ref{fig:cactus}, then the difference will be much more significant. If we have $t$ such cycles (for $t\geq 3$), then the algorithm of \cite{kkcw2025} uses the shortest path with length $k_1 = 4t-2$. The CNOT cost is $3(k_1-2)+2\cdot 2=3(4t-2-2)+4=12t-8$. The solution $P_2$ for the 1-SNSCP problem has $k_2=2t-2$ length. The CNOT cost is $3k_2+2(3t+1-k_2-1)=8t-2$.
\subsection{Circuit for Quantum Fourier Transform Algorithm}\label{sec:qft}
Let us consider a quantum device with some qubit connectivity graph $G=(V,E)$. We assume that $G$ is a connected graph. Here we present a method that allows us to construct a circuit that implements the Quantum Fourier Transform (QFT) algorithm on this device. More information on the QFT algorithm can be found in Appendix \ref{apx:qft}. If we do not have restrictions for applying two-qubit gates (when $G$ is a complete graph, for instance), then the circuit is presented in Figure \ref{fig:qft-cascade}.
We can split the circuit for the QFT algorithm into a series of control phase gates cascades depending on the target qubit for control phase operations. The $r$-th cascade uses $q_r$ as the target qubit (Figure \ref{fig:qft-cascade}).

\begin{figure}[!ht]
\includegraphics[width=0.45\textwidth]{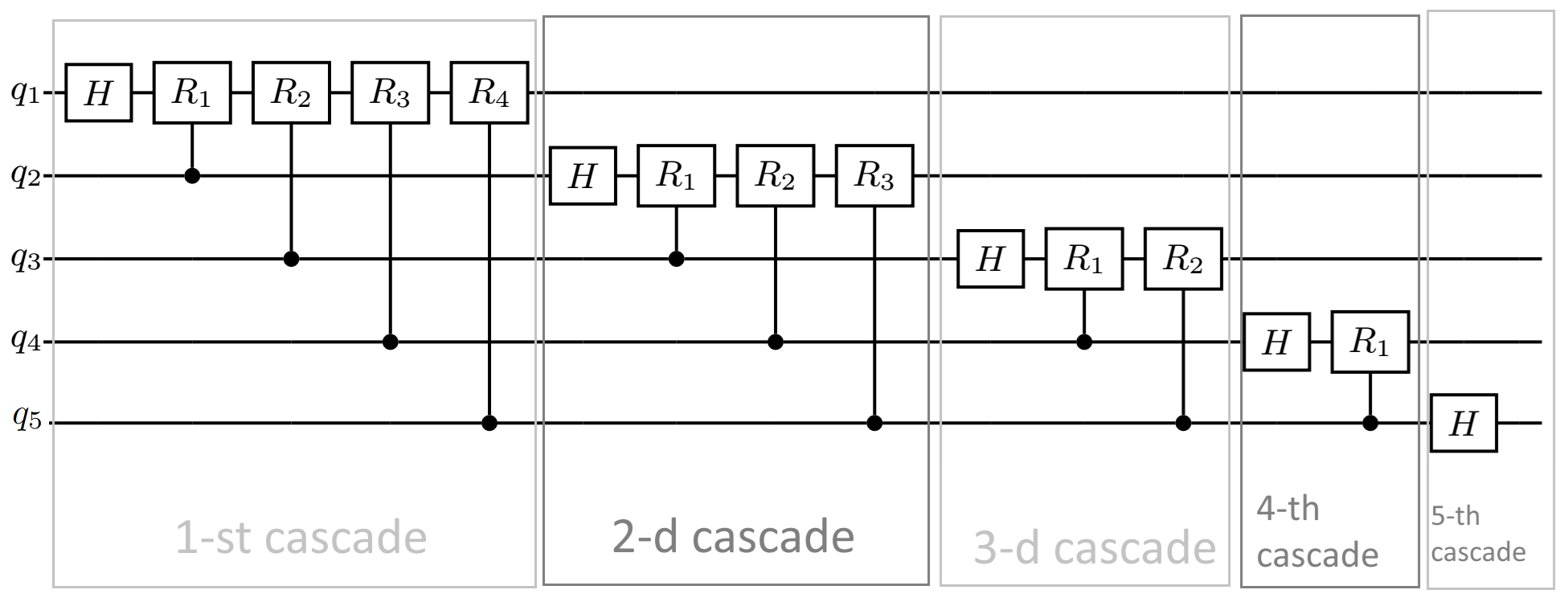}
\caption{\label{fig:qft-cascade} A quantum circuit for Quantum Fourier Transform algorithm for fully connected $5$ qubits splited to $5$ cascades depending on the target qubit.}
\end{figure}

We can see that each cascade has the structure same as a circuit for quantum hashing. That is why we can construct the circuit using similar method. There is an algorithm for constructing a circuit based on searching for a covering path \cite{kksk2026}. At the same time, the algorithm requires exponential time complexity. Based on our solution for a cactus from Section \ref{sec:tools}, we can construct a circuit in polynomial time. The algorithm from \cite{kksk2026} is described in Appendix \ref{apx:qft2}. The result is presented in the next theorem.

\begin{theorem}
    The algorithm of constructing a quantum circuit for QFT algorithm for cactus qubit connectivity graph has $O(n^3)$ time complexity. The constructed circuit has CNOT cost at most $2n^2$.
\end{theorem}
    The time complexity comes from Theorem \ref{th:wnshpath}. The CNOT cost is cost of the circuit build by algorithm from \cite{kksk2026}. 

\section{The Shortest Non-simple 1-covering Path Problem in a Cactus}\label{sec:tools}
Let us consider a cactus $G=(V,E)$ such that $n=|V|$ is a number of vertices and $m=|E|$ is a number of edges.
In this section, we discuss the ``Shortest non-simple 1-covering path'' problem that is a modification of the well-known ``Shortest covering path'' problem (SCPP or SCP problem) \cite{cpr1994}. The descriptions of the 1-SNSCP problem is presented in Section \ref{sec:prelims}.
As the SCP problem, the 1-SNSCP problem has a strong connection with the Hamiltonian path problem and the Travelling salesman problem \cite{cormen2001}. Any connected graph has a 1-covering path.
The decision version of the SCP problem is NP-complete \cite{cpr1994}. The Travelling Salesman Problem (TSP) is NP-hard. Similarly, by polynomial reduction of TSP to 1-SNSCPP, we can show that it is NP-hard.  

At the same time, in the case of a cactus, we suggest a polynomial-time algorithm. The solution has 
$O(n^3)$ 
time complexity (Theorem \ref{th:wnshpath}).  The results of this section can be interesting itself from the graph-theory point of view.
%
%
Let us estimate the maximum possible length of a 1-covering path.

\begin{lemma}\label{lm:len-wnsh}
The length of a 1-covering path in a connected graph $G$ of $n$ vertices is at most $2n-3$.  (See Appendix \ref{apx:len-wnsh}).
\end{lemma}

Let us present the algorithm for the 1-SNSCP problem for a cactus $G$.

Before starting the algorithm, we exclude all leafs because the path will note visit leafs. 
Firstly, we construct a weighted vertex cactus $T=(V',E')$ by $G$. Let $V_1\subset V$ be a set of vertices that belongs to at most one cycle, and $V_2=V\backslash V_1$ be a set of vertices that belongs to at least two cycles. Let $cyc(v)$ be the number of cycles whose $v$ belongs, where $v\in V_2$. Each vertex $v\in V_1$ corresponds to $v'\in V$. Each vertex $v\in V_2$ corresponds to $cyc(v)$ vertices $v'_1,\dots v'_{cyc(v)}\in V'$. There are no other vertices in $V'$ except the vertices listed above.

Each edge $(v,u)\in E$ such that $v,u\in V_1$ corresponds to an edge $(v',u')\in E'$, and weight $w(v',u')=1$.
Each edge $(v,u)\in E$ such that $u\in V_1$ and $v\in V_2$ corresponds to an edge $(v'_i,u')\in E'$, and weight $w(v',u')=1$. Here we assume that we can enumerate all cycles to which $v$ belongs and $v'_i$ belongs to $i$-th cycle. We can detect the index of the cycle by the vertex $u$. 
Each edge $(v,u)\in E$ such that $u,v\in V_2$ corresponds to an edge $(v'_i,u'_j)\in E'$, and weight $w(v',u')=1$. Were $v'_i$ belongs to $i$-th cycle in enumeration of $v$, and $u'_j$ belongs to the $j$-th cycle in enumeration of $u$ if these are the same cycles. 
We also add edges $(v'_1,v'_j)$ to $E'$ with weight $w(v'_1,v'_j)=0$ for each vertex $v\in V_2$ and the corresponding vertices $v'_1,\dots v'_{cyc(v)}\in V'$. There are no other edges in $E'$ except the edges listed above.

By any path $P'$ in $T$ we can form a path $P$ in $G$ by taking corresponding vertices. If $v'_i$ in the path $P'$, then we replace it with corresponding $v$. If there are two sequential vertices $v'_{i_1}$ and $v'_{i_2}$, then we replace it by one vertex $v$.

Let us show three properties of the new weighted vertex cactus $T$. We list them in the next lemmas.

\begin{lemma}
Let $P'$ be the $1$-covering path in $T$ with minimal weight. Then the path $P$ in $G$ formed by $P'$ is the shortest $1$-covering path and solution of the 1-SNSCP problem in $G$.
\end{lemma}
%
\begin{lemma}\label{lm:cactus-vertex}
For a cactus $G(V,E)$ and constructed $T(V',E')$, the following statements hold: $|V'|\leq 2|V|$, $|E'|\leq |E|+3|V|$ (See Appendix \ref{apx:cactus-vertex})
\end{lemma}
\begin{lemma}\label{lm:constr-T} The time complexity of constructing $T$ is $O(n+m)$. (See Appendix \ref{apx:constr-T}). 
\end{lemma}

Secondly, we develop an algorithm for a weighted vertex cactus.
Let a block of $T$ be a cycle or a vertex that does not belong to any cycle.
Let us construct a tree $A$ by a vertex cactus $T$ where each vertex of $A$ corresponds to a block of $T$. Two vertex in $A$ are connected if corresponding blocks in $T$ are connected. The weight of the edge in $A$ is the weight of the edge between corresponding blocks in $T$. See Figure \ref{fig:cactus01} for example.

\begin{figure}[!ht]
\begin{center}
\includegraphics[width=0.6\textwidth]{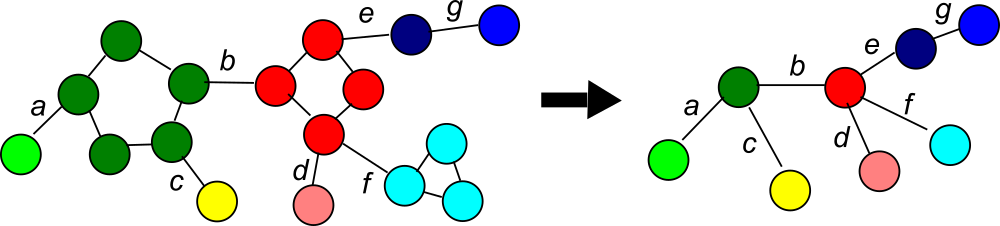}
\caption{\label{fig:cactus01} Construction of the tree A by the vertex cactus $T$.}    
\end{center}
\end{figure}

Let us fix a vertex $r$ as a root of the tree $A$. Assume that $r$ corresponds to a single vertex in $T$. We use the dynamic programming approach \cite{cormen2001} for developing an algorithm. Let us consider a vertex $v'$ in $T$ and the vertex $u$ of $A$ such that $v'$ belongs to the block corresponding to $u$. Let us compute two values $d_l[v']$ and $d_p[v']$ for each vertex of $T$. They are the lengths of the shortest paths that cover all the vertices of $T$ from  blocks corresponding to vertices of the subtree with root $u$ in the tree $A$.
\begin{itemize}
\item  $d_l[v']$ is the length of the path that starts and finishes in $v'$;
\item  $d_p[v']$ is the length of the path that starts in $v'$ and finishes somewhere in the subtree.
\end{itemize}


For a vertex $u$ of the tree $A$, we say that $d_l[u]=d_l[v']$ and $d_p[u]=d_p[v']$  if $u$ corresponds to a vertex $v'$. If $u$ corresponds to a cycle, then we choose $v'$ as the vertex of the cycle visited by the depth-search first (DFS) algorithm \cite{cormen2001} first if we start from $r$.

Let us consider a vertex $v'$ of $T$ and corresponding vertex $u$ of $A$. Let $p$ be a parent node of $u$, and $\textsc{Children}(u)$ be its children in $A$ if the root is the vertex $r$. If $u$ corresponds to a single vertex, then
\[d_l[v']=d_l[u] = \sum\limits_{i \in \textsc{Children}(u)} (d_l[i] + 2w(i,u))\] because we should visit all children subtrees and return to the vertex. We pass each edge that connects the vertex $u$ with a child $i$ two times (down and up).
\[d_p[v']= d_l[u] = \sum\limits_{i \in \textsc{Children}(u)} (d_l[i] + 2w(u,i)) - \max\limits_{i \in \textsc{Children}(u)} (d_l[i] - d_p[i] + w(u,i))\] because we should visit all children vertices and return to the vertices for each child except one. We choose this exceptional child that minimizes the result. See Figure \ref{fig:cactus02} for an example.
\begin{figure}[!ht]
\begin{center}
\includegraphics[width=0.5\textwidth]{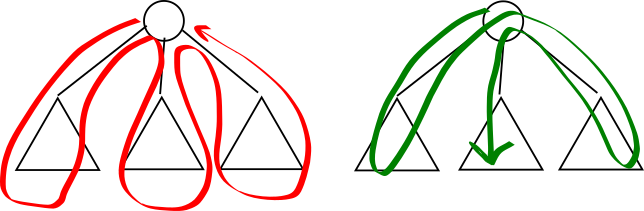}
\caption{\label{fig:cactus02} The left picture: computing $d_l[v']$ as a sum of $d_l[i]$  of children. The right picture: computing $d_p[v']$ as a sum of $d_l[i]$ of children except one child which $d_p[i]$ used.}    
\end{center}
\end{figure}

If $u$ corresponds to a cycle vertex, then we are interested in $d_l[v']=d_l[u]$, where $v'$ is the vertex that is the DFS algorithm started from $r$ visits first in the cycle.
$d_l[v']=d_l[u] = L_c+\sum\limits_{i \in \textsc{Children}(u)} (d_l[i] + 2w(i,u)),$ where $L_c$ is the minimum of two values 
\begin{itemize}
    \item  $\sum\limits_{(i,j)\mbox{ from the cycle of }u}w(i,j)$ and 
    \item $2\sum\limits_{(i,j)\mbox{ from the cycle of }u}w(i,j)-2\max\{w(i,j):(i,j)\mbox{ from the cycle of }u\}$.
\end{itemize}  Here, we should visit all the children subtrees, return to the cycle, and visit all vertices of the cycle. The $L_c$ is the length of the journey by the cycle. We can visit all vertices of the cycle by two ways: just go by cycle (that is the first option for $L_c$), and visit all edges except the largest one  (that is the second one). Note that for the second option, we should visit all edges twice: going forward and back. See Figure \ref{fig:cactus03} for an example.
\begin{figure}[!ht]
\begin{center}
\includegraphics[width=0.55\textwidth]{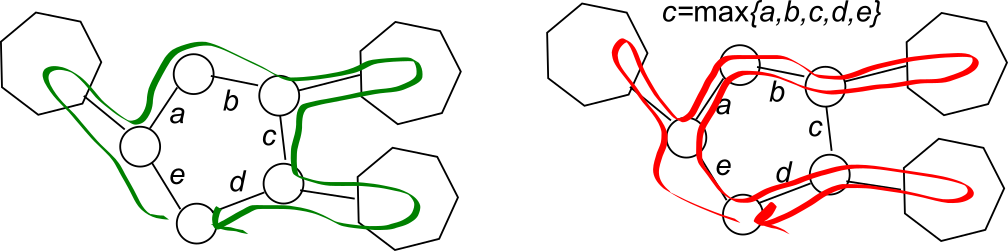}
\caption{\label{fig:cactus03} Choosing minimum from two options:just go by cycle and visit all edges except the largest one }    
\end{center}
\end{figure}

For computing $d_p[v']= d_l[u]$ we iterate by edges of the cycle corresponding to $u$, remove the edge and compute $d_j$ for $j$-th edge, where $d_j = \sum\limits_{i \in \textsc{Children}(u)} (d_l[i] + 2w(u,i)) - \max\limits_{i \in \textsc{Children}(u)} (d_l[i] - d_p[i] + w(u,i)).$ Here after removing the edge, we get several additional vertices in the tree $A$ each of them corresponds to a vertex of the cycle. For each removing edge we should recompute $d_l$ and $d_p$ values for these vertices. After the step we return the removed edge back. Finally, $d_p[v']= d_l[u] = \min\{ d_1\dots,d_t\}$, where $t$ is the size of the cycle. 

If the root $r$ corresponds to a cycle, then we try all vertices of the cycle and try to remove each edge of the cycle one by one and return it. For each of these case we obtain the tree where the root corresponds to a single node.

For a fixed root $r$ and a fixed removed edge from the cycle if $r$ belongs to a simple cycle, let $u$ be the vertex of $A$ corresponded to $r$. We choose two different $i,j\in \textsc{children}(u)$ such that $d_l[i]-d_p[i]+d_l[j]-d_p[j]$ is maximum. Then, the result is $d_l[r]-d_l[i]-d_p[i]+d_l[j]-d_p[j]$ because we should start from one of the children and finish in another. The value $d_l[r]$  already contains $d_l[i]$ and $d_l[j]$, that means we enter to the $i$-th and $j$-th children and com back to $r$. At the same time we want to start from $j$-th (we should use $d_p[j]$ in the sum but not $d_l[j]$), and finish in $i$-th (we should use $d_p[i]$ in the sum but not $d_l[i]$). See Figure \ref{fig:cactus04} for an example.
\begin{figure}[!ht]
\begin{center}
\includegraphics[width=0.3\textwidth]{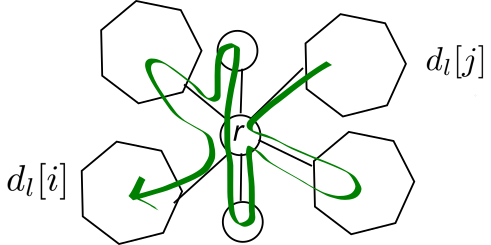}
\caption{\label{fig:cactus04} Starting from $j$-th child, go throw all other children except $i$-th with coming back and finishing in $i$-th child.}    
\end{center}
\end{figure}

We choose the minimal result for all possible roots $r$ and removed edges of the cycle if $r$ belongs to a simple cycle.

Let us estimate the complexity of the presented algorithm.

\begin{theorem}\label{th:wnshpath}
    The presented algorithm solves the 1-SNSCP problem for a cactus, and the time complexity is $O(n^3)$. (See Appendix \ref{apx:twnshpath})
\end{theorem}

\section{Conclusion}\label{sec:conclude}
In this paper, we present an algorithm for constructing a shallow quantum circuit for the quantum hashing algorithm for a device with a specific qubit connectivity graph. The algorithm uses solution of the 1-SNSCP problem as the main part. For a cactus graphs, our result gives better results than the results of the existing method for circuit constructing for arbitrary qubit connectivity graphs \cite{kkcw2025}. Additionally, our algorithm has polynomial time complexity, and the algorithm of  \cite{kkcw2025} has exponential time complexity. The algorithm for 1-SNSCP in the case of a cactus graph is interesting from the graph theory point of view since it has a strong connection with the Shortest covering path \cite{cpr1994}.
\bibliographystyle{unsrt}
\bibliography{tcs}
\newpage
\appendix
\section{Quantum Fingerprinting or Quantum Hashing}\label{apx:hash}
Let us present some basic concepts of quantum fingerprinting technique from \cite{af98,an2008,an2009,bcwd2001}. This technique is well-known and allows us to compute a short hash or fingerprint that identifies a (potentially long) string of data with high probability.

For the problem being solved we choose some positive integer $m$, an error probability bound $\varepsilon > 0$, fix $t = \lceil(2/\varepsilon) \ln 2m\rceil$, and construct a mapping $g : \{0, 1\}^n\to \mathbb{Z}$. Then for arbitrary binary string $\sigma = (\sigma_1 \dots \sigma_n)$ we create it's fingerprint $|h_\sigma\rangle$ composing $t$ single qubit fingerprints $|h_\sigma^i\rangle$:
\[|h_\sigma^i\rangle=\cos\frac{2\pi k_i g(\sigma)}{m}|0\rangle + \sin\frac{2\pi k_i g(\sigma)}{m}|0\rangle,\]\[
|h_\sigma\rangle=\frac{1}{\sqrt{t}}\sum_{i=1}^{t}|i\rangle|h^i_{\sigma}\rangle\]

Here the last qubit is rotated by $t$ different angles about the $\hat{y}$ axis of the Bloch sphere. The chosen parameters $k_i \in\{1\dots,m-1\}$, for $i\in\{1\dots t\}$ are ``good'' in the following sense. A set of parameters $K = \{k_1,\dots, k_t\}$ is called ``good'' for $g\neq 0 \mod m$ if
\[\frac{1}{t^2}\left(\sum_{i=1}^t \cos\frac{2\pi k_i g}{m}\right)^2<\varepsilon\]
The left side of the inequality is the squared amplitude of the basis state $|0\rangle^{\otimes \log_2 t} |0\rangle$ if the operator
$H^{\otimes \log_2 t}\otimes I $ has been applied to the fingerprint $|h_\sigma\rangle$. Informally, that kind of set guarantees, that
the probability of error will be bounded by a constant below $1$.

The following lemma from \cite{akv2008,an2008,an2009} proves the existence of a ``good'' set.
\begin{lemma}[\cite{akv2008}]
There is a set $K$ with $|K| = t = \lceil(2/\varepsilon) \ln 2m\rceil$  which is ``good'' for all $g\neq 0 \mod m$.
\end{lemma}

We use this result for fingerprinting technique \cite{akv2008} choosing the set $K = \{k_1,\dots, k_t\}$ that is ``good'' for all $g = g(\sigma)\neq 0$. It allows us to distinguish those inputs whose image is $0$ modulo $m$ from the others.

That hints at how this technique may be applied:
\begin{enumerate}
\item We construct $g(x)$, that maps all acceptable inputs to $0$ modulo $m$ and others to arbitrary non-zero (modulo $m$) integers.

\item After the necessary manipulations with the fingerprint, the $H^{\otimes \log_2 t}$ operator is applied to the first $\log_2 t$ qubits. This operation ``collects'' all cosine amplitudes at the all-zero state. That is, we obtain the state of the type
\[|h'_\sigma\rangle=\frac{1}{t}\sum_{i=1}^{t}\cos\left(\frac{2\pi k_i g(\sigma)}{m}\right) |00\dots 0\rangle|0\rangle + \sum_{i=2}^{2t}\alpha_i|i\rangle\]
\item This state is measured on the standard computational basis. Then we accept the input if the outcome is the all-zero state. This happens with the probability
\[Pr_{accept}(\sigma)=\frac{1}{t^2} \left(\sum_{i=1}^{t}\cos\frac{2\pi k_i g(\sigma)}{m}\right)^2,\]
which is $1$ for the inputs, whose image is $0 \mod m$ and is bounded by $\varepsilon$ for the others.
\end{enumerate}

Due to \cite{ziiatdinov2023gaps,zkk2025,kalis18}, the algorithm can be implemented using the shallow circuit presented in Figure \ref{fig:qf}. In that case, the angles $\frac{2\pi k_i}{m}$ should be linear combinations of $\xi_j$. Due to \cite{ziiatdinov2023gaps,zkk2025}, it is not known whether we can keep the same number of qubits or should we increase the number of qubits exponentially. At the same time, computational experiments show \cite{ziiatdinov2023gaps,zkk2025} that we can find enough good parameters $\xi_i$ such that $t$ qubits are enough for $\varepsilon$ error probability.

\subsection{The Quantum Finite Automaton for $MOD_p$ Language}\label{sec:qfamodp}

The general description of quantum hashing or quantum fingerprinting technique from \cite{af98,an2008,an2009,akv2008,bcwd2001} can be found in Appendix \ref{apx:hash}.  This technique is well-known and allows us to compute a short hash or fingerprint that identifies a (potentially long) string of data with high probability.
Here we present a quantum finite automaton for the unary $MOD_p=\{a^\ell: \ell$ mod $p=0$, $p$ is prime$\}$ language as a simple example of this technique. At the same time, it shows the the power and have all the main properties of the technique. 

A quantum finite automaton (QFA) for unary language is a tuple $(Q, Q_{acc}, Q_0, U_a, U_{\$}, U_{\cent})$. Here $Q$ is a set of $2^n$ quantum basis states ($n$ qubits), $Q_{acc}\subset Q$ is the set of accepting states, $Q_0\in Q$ is a starting state, and $U_a, U_{\$}$ and $ U_{\cent}$ are unitary transformations on $2^n$ states.

The automaton processes $\$ \omega \cent$ word for input word $\omega=a^\ell$ for some positive integer $\ell$, $\$$ is the  left end marker, and $\cent$ is the right end marker. The automaton works with $n$-qubit register, starts from the $|Q_0\rangle$ state. We apply $U_{\$}$ to the register on left end marker, then we apply $U_{a}$ on each symbol of $w$, and finally, apply  $U_{\cent}$ on the right end marker. Finally, we measure the quantum register and accept the word if the state belongs to $Q_{acc}$. The probability of this event is $\sum_{Q_i\in Q_{acc}}|\alpha_i|^2$, where $\alpha_i$ is the amplitude of $i$-th state. A language is recognizable by an automaton if the automaton accepts member words with probability $1-\varepsilon$, and rejects non-members with probability $1-\varepsilon$ for some $0<\varepsilon<0.5$. 
The automaton for $MOD_p$ language \cite{an2008,an2009} is such that $n=\log_2((2/\varepsilon) \ln 2p)$, $|Q_0\rangle=|0\rangle$, $Q_{acc}=\{|0\rangle\}$.
The transformations $U_{\$}$ and $U_{\cent}$ are Hadamard transformations $H$ on each qubit except the $0$-th one. $U_a$ is the $UCR_y^{n-1}$ gate with angles $\xi_i=\frac{2\pi k_i}{p}$ for $k_0,\dots, k_{2^{n-1}-1}\in\{1,\dots,p-1\}$. Here $k_i$ are good coefficients from \cite{an2008} discussed in the previous section.

Note that  $p$ can be any integer. At the same time, for prime $p$, the quantum automata can use exponentially less memory comparing to classical counterparts \cite{af98}. There are results \cite{AY12,agky16} that use fingerprinting for non-prime $p$. 


\section{Implementation of the Algorithm for Constructing a Quantum Circuit}\label{apx:algo-impl}
Let us present a procedure that implements the algorithm discribed in Section \ref{sec:hash}. We assume that we have $\textsc{cR}(u,v)$ procedure that applies the control rotation operator $CR_y$ to $u$ as a control qubit and $v$ as a target one. The angle $\xi_r$ corresponds to the qubit $q_r$ associated with the vertex $v$. Additionally, we have $\textsc{swap}(u,v)$ procedure that applies swap gate to $u$ and $v$ qubits.

\begin{algorithm}[H]
\caption{Implementation of $\textsc{ConstructForPath}(P)$ procedure. Algorithm of constructing circuit for quantum hashing or quantum fingerprinting for a path $P=(v_{i_1}, \dots,v_{i_k})$}\label{alg:cascade}
\begin{algorithmic}
\State $P=(v_{i_1}, \dots,v_{i_k})\gets \textsc{Shortest1CP}(G)$\Comment{Step 1}
\State $U\gets\emptyset$\Comment{Step 2}
\For{$j\in\{1,\dots,k-1\}$}
\For{$v'\in\textsc{Neighbors}(v_{i_j})$}\Comment{Step 3}
\If{$v'\not\in U$ and $v'\neq v_{i_{j+1}}$}
\State  $\textsc{cR}(v',v_{i_j})$
\State $U\gets U\cup\{v'\}$
\EndIf
\EndFor
\If{$v_{i_{j+1}}\not\in U$}
\State  $\textsc{cR}
(v_{i_{j+1}},v_{i_j})$\Comment{Step 4}
\State $U\gets U\cup\{v_{i_{j+1}}\}$
\EndIf
\State $\textsc{swap}(v_{j},v_{j+1})$\Comment{Step 5}
\EndFor
\For{$v'\in\textsc{Neighbors}(v_{i_k})$}\Comment{Step 6}
\If{$v'\not\in U$}
\State  $\textsc{cR}(v',v_{i_k})$
\State $U\gets U\cup\{v'\}$
\EndIf
\EndFor
\end{algorithmic}
\end{algorithm}

\section{The Proof of the Theorem \ref{th:qh1}}\label{apx:qh1}

\textbf{Theorem \ref{th:qh1}} {\em
    The CNOT cost of the circuit for $\ell$ applications of quantum hashing operator $U_s$ generated by the presented algorithm is $(3k + 2(n-k'))\ell-5\ell + 2$, where $k$ is the length of the 1-covering path $P=(v_{i_1},\dots,v_{i_k})$, and $k'$ is the number of distinct vertices in $P$, and $n$ is the number of vertices in the qubit connectivity graph. 
}
\begin{proof}
If we look at the algorithm, then we can see that it is a sequence of one of three options:
    (i) a pair of $CR_y$ and $SWAP$ gates if $v_{j+1}\not\in U$;
    (ii) $SWAP$ gate if $v_{j+1}\in U$;
    (iii) $CR_y$ gate for  $v'$ and $v_{i_{j+1}}$.
%
The CNOT cost for pairs of $CR_y$ and $SWAP$ gates is $3$; for the $SWAP$ gate, it is $3$; for the $CR_y$ gate, it is $2$.
On applying $U_s$ once, the number of steps in which we apply pairs of gates $CR_y$ and SWAP or the SWAP gate is $k-1$ because we move the target qubit by the path of length $k$. We apply a single  $CR_y$ (without SWAP gate) to vertices that do not belong to $P$. The number of such nodes is $n-k'$.
At the same time, for two sequential applications of the $U_s$ operator, one $CR_y$ operator disappears, as we discussed before.
So, the CNOT cost of the first $\ell-1$ operators $U_s$ is $(3(k-1) + 2(n-k'-1))(\ell - 1)$, and the CNOT cost of the last operator is $3(k-1) + 2(n-k')$. The total cost is $(3k + 2(n-k'))\ell-5\ell + 2$.
\Endproof\end{proof}

\section{The Proof of the Corollary \ref{cr:path}}\label{apx:path}

\textbf{Corollary \ref{cr:path}} {\em
The CNOT cost of the circuit for $\ell$ applications of the quantum hashing operator $U_s$ generated by the presented algorithm is between $2n\ell-4\ell+2$ and $6n\ell-7\ell+2$, where $k$ is the length of the 1-covering path $P=(v_{i_1},\dots,v_{i_k})$, and $k'$ is the number of distinct vertices in $P$. 
}
\begin{proof}
The minimum possible CNOT cost is in the case when the length of $P$ is minimal. If we have a vertex that is connected with all other vertices of the graph, then we can construct $P$ of length $1$. So, $k=1$ and $k'=1$. Therefore, the total CNOT cost is $2n\ell-4\ell+2$.

Due to Lemma \ref{lm:len-wnsh}, the maximum length of the path is $k=2n-2$. The maximum possible number of distinct vertex in $P$ is $n-2$. So, the total CNOT cost is $(3(2n-2)+4)\ell-5\ell+2=6n\ell-7\ell+2$.
\Endproof\end{proof}
\section{Proof of Lemma \ref{lm:len-wnsh}.}\label{apx:len-wnsh}
\textbf{Lemma 
 \ref{lm:len-wnsh}} {\em
The length of a 1-covering path in a connected graph $G$ of $n$ vertices is at most $2n-3$. 
}

\begin{proof}
Let us consider a spanning tree of the graph $G=(V,E)$. It is a tree $T=(V,E')$, where $E'\subset E$. 
We can construct a non-simple path $P$ that is the Euler tour \cite{cormen2001} of the tree $T$ but does not visit the leaves of the tree. The path covers all the vertices of the graph $G$, but maybe be it is not the shortest. Each edge (except edges incident to leafs)  in the tour is visited at most twice (in the up and down direction). Therefore, the length of the path $len(P)\leq 2n-{\cal L}$, where ${\cal L}$ is the number of leaves, and ${\cal L}\geq 2$. So, we obtain the bound for the number of vertices in the path $2n-2$, and for the length of the path, the bound is $2n-3$.
\Endproof\end{proof}
\section{The Proof of Lemma \ref{lm:cactus-vertex}}\label{apx:cactus-vertex}
\textbf{Lemma 
 \ref{lm:cactus-vertex}} {\em  
For a cactus $G(V,E)$ and constructed $T(V',E')$, the following statements hold: $|V'|\leq 2|V|$, $|E'|\leq |E|+3|V|$
}

\begin{proof}
    Let us consider a tree $B$, where each vertex corresponds to a block of $G$. Here, a block is a cycle or a single vertex that does not belong to any cycle. Two vertices of $B$ are connected if the corresponding blocks are connected by an edge or have a common vertex. Due to properties of a cactus, the graph $B$ is a tree. Therefore, we have at most $h$ edges in the tree, where $h$ is the number of blocks. Note that $h\leq |V|$. It is easy to see that each edge in $B$ corresponds to a connection $(v'_1,v'_j)$ in $T$. Hence, the number of new additional vertices in $T$ is at most $h\leq |V|$. So, we have $|V'|\leq|V| + h\leq 2|V|$.

    Each additional vertex $v'_j$ produces at most three additional edge - one for connecting with $v'_1$, and two for connecting with two vertices of the $j$-th cycle. As we show before, the number of additional vertices is at most $|V|$. Therefore, we have $|E'|\leq 3|V|+|E|$.
 \Endproof
\end{proof}
\section{Proof of Lemma \ref{lm:constr-T}.}\label{apx:constr-T}
\textbf{Lemma 
 \ref{lm:constr-T}} {\em  The time complexity of constructing $T$ is $O(n+m)$. 
}

\begin{proof}
Firstly, using the depth-first search (DFS) algorithm \cite{cormen2001} we find each cycle. The complexity of the algorithm is $O(n+m)$. After that we form the sets $V_2$ and $V_1$. Then, we can form the graph $T$ in linear time. We should add  at most $3n+m$ edges and at most $n$ new vertices. The total complexity is $O(n+m)$.
    \Endproof
\end{proof}

\section{Proof of Theorem \ref{th:wnshpath}.}\label{apx:twnshpath}
\textbf{Theorem 
 \ref{th:wnshpath}} {\em  The presented algorithm solves the 1-SNSCP problem for a cactus, and the time complexity is $O(n^3)$.
}

\begin{proof}
    Let us fix a root vertex $r$. We use the DFS algorithm that has $O(|V'|+|E'|)$ time complexity. At the same time, if we consider a block that is a cycle, it takes $O(t)$ steps, where $t$ is the size of the cycle. In each step, we try to remove one edge of the cycle. In each step, we recompute $d_l$ and $d_p$ for the vertices of the cycle, but other values are not recomputed. So, processing a cycle of size $t$ has $O(t^2)$ time complexity. We process each cycle in a similar way. If the graph $T$ contains $h$ cycles of sizes $t_1,\dots, t_h$, then the total complexity of the algorithm for a fixed root vertex is $O(|V'|+|E'|+t_1^2+\dots+t_h^2)=O(|V'|+|E'|+(t_1+\dots+t_h)^2)=O(|V'|+|E'|+(|V'|)^2)=O(|V'|^2)$.
    Note that $|E'|\leq |V'|^2$.

    We iterate different vertices as a root. Let us consider a vertex $u$ of the tree $A$ that corresponds to a cycle in $T$. Let the size of the cycle is $t$. If we fix one of the vertex of the cycle as a root and try to remove different edges, then we should recompute only $d_l$ and $d_p$ only for vertices of the current cycle, because the values for other vertices are not changed. Therefore, removing $t$ edges and recomputing $d_l$ and $d_p$ for each step has $O(t^2)$ time complexity. Then we try different vertices of the cycle as a root, and have the same situation. So, checking one cycle has $O(|V'|^2)$ time complexity for the first running DFS, and additional $O(t^3)$ time complexity for checking all other options for root vertex of the cycle and different removed edges. Therefore, the total time complexity for checking all vertices of cycle as a root is $O(|V'|^2+t^3)$.  If the graph $T$ contains $h$ cycles of sizes $t_1,\dots, t_h$, then the total complexity is $O(h|V'|^2+t_1^3+\dots+t_h^3)=O(|V'|^3+(t_1+\dots+t_h)^3)=O(|V'|^3+|V'|^3)=O(|V'|^3)=O(n^3)$ due to Lemma \ref{lm:cactus-vertex}. If we add complexity of constructing $T$ from Lemma \ref{lm:constr-T}, then we obtain $O(n^3+n+m)=O(n^3)$.
\Endproof\end{proof}  

\section{Quantum Fourier Transform}\label{apx:qft}
QFT is a quantum version of the discrete Fourier transform. The definitions of $n$-qubit QFT and its inverse are as follows:
\[QFT|j\rangle = \sum_{k=0}^{2^n-1}e^{\frac{2\pi i jk}{2^n}}|k\rangle,\]
\[QFT^{-1}|j\rangle = \sum_{k=0}^{2^n-1}e^{-\frac{-2\pi i jk}{2^n}}|k\rangle,\]
The $n$-qubit QFT circuit requires $0.5n^2 - 0.5n$ control phase ($CR_d$) gates and $n$ Hadamard ($H$) gates if we have no restriction on the application of two-qubit gates (See Figure \ref{fig:cqft}).

\begin{figure}[H]
\includegraphics[width=0.45\textwidth]{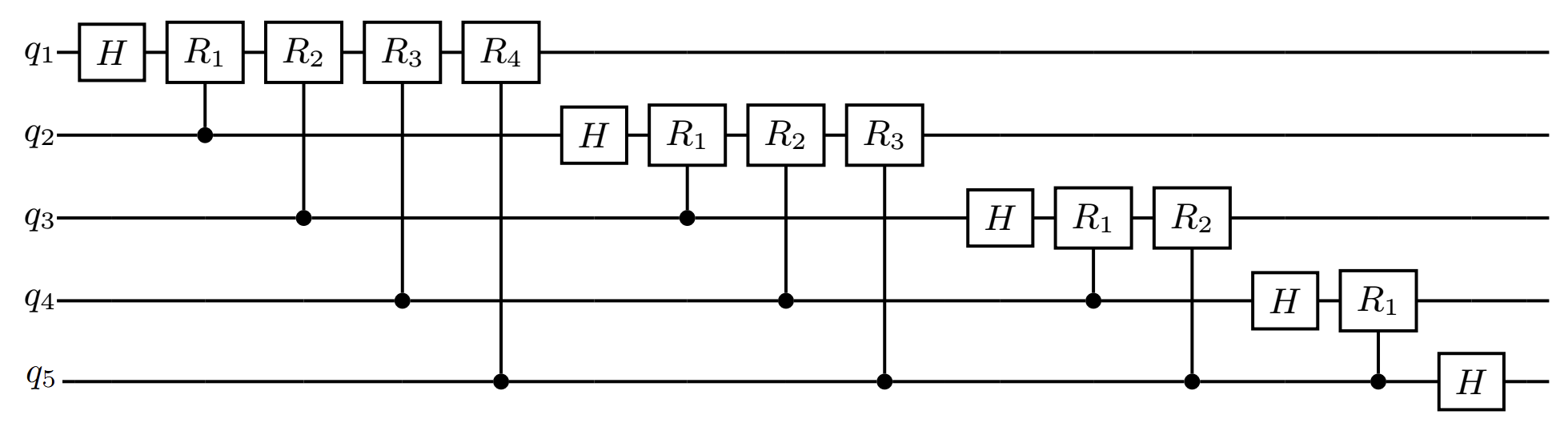}
\caption{\label{fig:cqft} A quantum circuit for Quantum Fourier Transform algorithm for fully connected $5$ qubits}
\end{figure}

The $CR_d$ gate is represented by basic gates that require two CNOT and three $R_z$ gates \cite{barenco1995elementary}. Therefore, $n^2 - n$ CNOT gates are required to construct an $n$-qubit QFT circuit. At the same time, if a quantum device has the LNN architecture, then for implementing the QFT, the number of CNOT gates is much larger than $n^2 - n$ \cite{park2023reducing}. If we consider a general graph, then the situation is much worse than \cite{k2024aliya}.

\section{Implementation of the Circuit for Quantum Fourier Transform}\label{apx:qft2}

Let us consider a quantum device with some qubit connectivity graph $G=(V,E)$. We assume that $G$ is a connected graph. Here we present a method that allows us to construct a circuit that implements the Quantum Fourier Transform (QFT) algorithm on this device. More information on the QFT algorithm can be found in Appendix \ref{apx:qft}. If we do not have restrictions for applying two-qubit gates (when $G$ is a complete graph, for instance), then the circuit is presented in Figure \ref{fig:cqft}.

We can split the circuit for the QFT algorithm into a series of control phase gates cascades depending on the target qubit for control phase operations. The $r$-th cascade uses $q_r$ as the target qubit (Figure \ref{fig:qft-cascade}).

Assume that we have a $\textsc{CascadeForPath}(P_r,r)$ procedure that constructs the $r$-th cascade of the circuit for the QFT algorithm for a path $P_r$. 
Here $P_r$ is a path that ``covers'' only vertices corresponding to the qubits used in the current cascade. We say that a path covers a vertex if the vertex is visited by the path or the vertex is connected by an edge with some vertex from the path. Because we can apply two-qubit gates only for adjacent vertices, the procedure moves the target qubit by the path $P_r$ from the first vertex of the $P_r$ to the last one. We move the target qubit using the SWAP gate. During the ``travel'' of the target qubit, we apply the control phase operator to each neighbor vertex because the path $P_r$ covers all the vertices that correspond to the cascade. This strategy allows us to implement the cascade. In the end of the ``travel'', we move the target qubit to one of the neighbors of the last vertex of $P_r$ and exclude it from the next steps because it does not participate in rest cascades.  

Firstly, we present the main algorithm in Section \ref{sec:qft1}. Then we present the detailed algorithm for the $\textsc{CascadeForPath}(P_r,r)$ procedure in Section \ref{sec:qft2}. After that we discuss the complexity of the circuit in Section \ref{sec:qft3}.

\subsection{The Main Algorithm}\label{sec:qft1}
Let us present the entire algorithm for constructing the quantum circuit for the QFT algorithm. 

\subsubsection{Vertices and Qubits Correspondence}
Firstly, we should assign original qubits to the vertices. Consider two sequences:
\begin{itemize}
    \item $A_1,\dots, A_n$ are the indexes of initial positions of qubits. If $A_i=j$ on some step, it means that the vertex $v_i$ contains an original qubit that was in $v_j$ before starting the algorithm.
    \item $S_1,\dots S_n$ are the final positions of the qubits. If $S_i=j$, then the $j$-th  original qubit is located in the vertex $v_i$ before starting the algorithm.
\end{itemize}

Our main goal is to compute the sequence $S_1,\dots S_n$. Let us present the algorithm.

\begin{itemize}
    \item[] \textbf{Step 0.} We assign $A_i\gets i$ for each $i\in\{1,\dots,n\}$. Let $r\gets 1$ be the number of a cascade.
    \item[] \textbf{Step 1.} We find a the shortest covering path $P_r=(v_{i_1},\dots,v_{i_k})$.
    \item[] \textbf{Step 2.} We assign $S_{A_{i_1}}\gets r$ 
    \item[] \textbf{Step 3.} We move the first element by the path, i.e. we swap $A_{i_j}$ and $ A_{i_{j+1}}$ for $j\in\{1,\dots,k-1\}$.
    \item[] \textbf{Step 4.} We choose a neighbor vertex $v_q$ of $v_{i_k}$ with the maximal index that is not visited by the path $P_r$. Then we assign $A_{i_k}\gets A_{q}$, and we exclude the vertex $v_q$ from the graph\footnote{In fact, we do not exclude it, but mark as excluded. After invocation of this algorithm, we should be able to restore the whole graph.}.
    \item[] \textbf{Step 5.} We go to the next cascade $r\gets r+1$. If $r\leq n-2$, then we go to Step 1, and go to Step 6 otherwise.
    \item[] \textbf{Step 6.} In this step, we have two vertices in the graph that are not excluded and connected. Assume that there are $v_q$ and $v_t$, and $q<t$. Then, we assign $S_{A_q}\gets n-1$, and $S_{A_t}\gets n$.  
\end{itemize}

Assume that the $\textsc{ConstructS}(G)$ procedure contains it.

\subsubsection{The Algorithm}
The enumeration $S$ is such that the algorithm works well, and the algorithm for computing $S$ is very similar to the main algorithm.

First, we restore the graph. Then, on each cascade, the $r$-th original qubit is located at the starting vertex of the path $P_r$. For each cascade, we move the $r$-th original qubit by the path $P_r$ using the SWAP gate and then to the neighbor of the last vertex of the path with the maximal index. After that, we exclude the qubit from the graph. 

We use $Q_i$ as the current position of the $i$-th original qubit and $T_j$ as an index of an original qubit located in the vertex $v_j$. Initially $T_j\gets S_j$, $Q_{T_j}\gets j$ for each $j\in\{1,\dots, n\}$.

The construction of a cascade is presented by the procedure $\textsc{CascadeForPath}(P_r,r)$. The algorithm is as follows.

\begin{itemize}
    \item[] \textbf{Step 0.} We associate the $S_j$-th original qubit with the vertex $v_j$, i.e. $T_j\gets S_j$, $Q_{T_j}\gets j$, for $j\in\{1,\dots,n\}$. 
    
    Let $r\gets 1$ be the number of a cascade.
    
    \item[] \textbf{Step 1.} We construct the $r$-the cascade using $\textsc{CascadeForPath}(P_r,r)$ and keep the $T$ and $Q$ indexes actual.
    \item[] \textbf{Step 2.} We choose a neighbor vertex $v_q$ of $v_{i_k}$ with the maximal index that is not visited by the path $P_r$ and exclude it because the $r$-th qubit was moved there during the $\textsc{CascadeForPath}(P_r,r)$ procedure.
    \item[] \textbf{Step 3.} We go to the next cascade $r\gets r+1$. If $r\leq n$, then we go to Step 1, and stop otherwise.
\end{itemize}

\subsection{Quantum Circuit for One Cascade}\label{sec:qft2}

Let us present the algorithm for generating a quantum circuit for the $r$-th cascade, that is the procedure $\textsc{CascadeForPath}(P_r,r)$.

In the $r$-th cascade, we use the $r$-th qubit as a target for the control phase gates. Due to the enumeration of qubits, it is located in the vertex $v_{i_1}$, where $P_r=(i_1,\dots,i_k)$. 

We move the target qubit by the path $P_r$ and for each position of the target qubit, we apply control phase gates for each neighbor vertex. Finally, we move the target qubit to the neighbor of $v_{i_k}$ with the maximal index. For refusing repetition of applying of a control phase gate for a control qubit, we use a set $U$ that stores all qubits that have already been used as control qubits during this cascade.

The algorithm for constructing a quantum circuit is as follows.

\begin{itemize}
    \item[] \textbf{Step 1.} We start with the first qubit in the path $j\gets 1$, and initialize $U\gets\emptyset$. We apply the Hadamard transformation to the qubit corresponding to the vertex $v_{i_1}$. We denote this action by $\textsc{H}(v_{i_1})$. If $k=1$, then we terminate our algorithm; otherwise, go to Step 2. 
      \item[]  \textbf{Step 2.} For each $v_t\in \textsc{Neighbors}(v_{i_j})\backslash\{v_{i_{j+1}}\}$, if $v_t\not\in U$, then we apply the control phase gate $CR_d$ with the control $v_{t}$ and the target $v_{i_j}$ qubits, where $d=T_{t}-r$. Note that $v_t$ with the maximal index should be processed in the end.      
      Then, we add $v_{t}$ to the set $U$, i.e. $U\gets U\cup\{v_{t}\}$. If $j=k$, then we go to Step 5, and to Step 3 otherwise.
        \item[]  \textbf{Step 3.} If $v_{i_{j+1}}\not\in U$, then we apply the control phase gate $CR_d$ with the control $v_{i_{j+1}}$ and the target $v_{i_j}$ qubits, where $d=T_{i_{j+1}}-r$. Then, we add $v_{i_{j+1}}$ to the set $U$, i.e. $U\gets U\cup\{v_{i_{j+1}}\}$. After that, we go to Step 4.
     \item[]  \textbf{Step 4.} We apply the SWAP gate to $v_{i_{j}}$ and $v_{i_{j+1}}$, and swap the indexes of qubits for these vertices. In other words, if $w_1=T_{i_j}$ and $w_2=T_{i_{j+1}}$ are indexes of the corresponding original qubits, then we swap $Q_{w_1}$ and $Q_{w_2}$ values, and $T_{i_j}$ and $T_{i_{j+1}}$ values. Then, we update $j\gets j+1$ because the value of the target qubit moves to $v_{i_{j+1}}$. Then, we go to Step 2.
     \item[] \textbf{Step 5.} If $j=k$, then we apply the SWAP gate to $v_{i_{j}}$ and $v_{q}$, and swap the qubit indexes for these vertices similarly to Step 4. Here $v_q$ is the neighbor of $v_{i_{j}}$ with the maximal index, i.e. $q=\max\{j:$ $v_j$ is not excluded,$ v_j\in \textsc{Neighbors}(v_{i_k}), j\neq i_{k-1}\}$
\end{itemize}

Finally, we obtain the $\textsc{CascadeForPath}(P_r,r)$ procedure whose implementation is presented in Algorithm \ref{alg:qft-path} (see Appendix \ref{apx:qft-path}). This procedure constructs the $r$-th part (cascade) of the circuit for QFT for the path $P_r$.

\subsection{The CNOT cost of the Circuit}\label{sec:qft3}
Note that the $CR_d$ gate can be represented using only two CNOT gates and three $R_z$ gates \cite{barenco1995elementary} (see Figure \ref{fig:cp}).

\begin{figure}[H]
\includegraphics[width=0.4\textwidth]{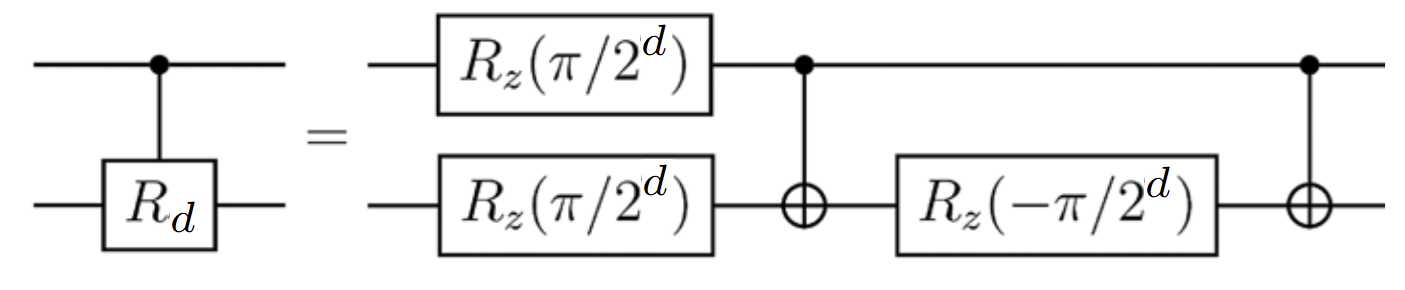}
\caption{\label{fig:cp} Representation of $CR_d$ gate using only basic gates}
\end{figure}  
A pair of $CR_d$ and $SWAP$ gates can be represented using three CNOT gates (see Figure \ref{fig:crswap}).

\begin{figure}[H]
\includegraphics[width=0.4\textwidth]{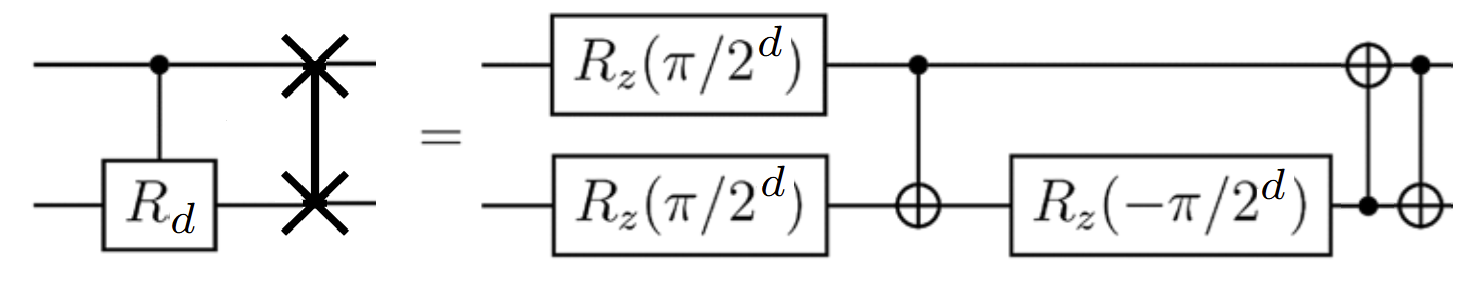}
\caption{\label{fig:crswap} Reduced representation of a pair $CR_d$ and $SWAP$ gates using only basic gates}
\end{figure}  

Let us discuss the CNOT cost of the algorithm in the next theorem.
\begin{theorem}[\cite{kksk2026}]\label{th:qft}
    The CNOT cost of the circuit that is generated using the presented algorithm  is at most $K+n^2-n-1$, where $K=\sum_{r=1}^{n-1}len(P_r)$ is the sum of lengths of the the shortest covering paths $P_r$.
\end{theorem}



We have two corollaries from this result. Firstly, we can estimate $K$ as $nk - 0.5k^2+1.5k$, where $k$ is the length of a the shortest covering path in the graph $G$. We present this result in Corollary $\ref{cr:maxK}$. Then, we obtain the minimal and maximal bounds for the CNOT cost in Corollary \ref{cr:bounds}. 

\begin{corollary}[\cite{kksk2026}]\label{cr:maxK}
 The CNOT cost of the circuit that is generated using the presented algorithm  is at most $nk-0.5k^2-1.5k+n^2-n$, where $k$ is the length of a the shortest covering path in the graph $G$. 
\end{corollary}

\begin{corollary}[\cite{kksk2026}]\label{cr:bounds}
The CNOT cost of a circuit that is generated usingthe presented algorithm  is in the range between $n^2-2n-2$ and $2n^2-2n-2$. 
\end{corollary}
\end{document}